\documentclass{llncs}
\usepackage{wrapfig}
\usepackage{amssymb,amsmath}
\usepackage{array}  
\usepackage{tikz}
\usepackage{tkz-graph}
\usetikzlibrary{snakes,shapes,backgrounds,arrows,matrix,petri} 
\usepackage{extarrows}
\usepackage{mathrsfs}
\usepackage{algorithm2e}

\newcommand{\BB}{\mathbb{B}}
\def\un{\ensuremath{\mbox{1\hspace{-.2em}l}}}

\newcommand{\zero}{\textsf{\textbf{0}}}
\newcommand{\one}{\textsf{\textbf{1}}}

\renewcommand{\bar}[1]{\overline{#1}}
\newcommand{\trans}[1][]{
	\hspace{0.3pt}\raisebox{.5ex}{
		\begin{tikzpicture}[descr/.style={fill=white,inner sep=2.5pt}]
			\path (0,0) edge[->, >=angle 60] node[above]
			{\ensuremath{{\scriptstyle #1}}} (0.7,0);
		\end{tikzpicture}
	}
}

\newcommand{\Ctwo}{\ensuremath{\mathscr{C}^r}}
\newcommand{\Cone}{\ensuremath{\mathscr{C}^\ell}}
\newcommand{\ctwo}{\ensuremath{c^r}}
\newcommand{\cone}{\ensuremath{c^\ell}}
\newcommand{\Dp}{\ensuremath{\mathscr{D}^+}}
\newcommand{\Dm}{\ensuremath{\mathscr{D}^-}}
\newcommand{\Dpm}{\ensuremath{\mathscr{D}^\pm}}
\newcommand{\com}{\ensuremath{c}}
\newcommand{\cop}[2]{\texttt{copy}(#1, #2)}
\newcommand{\pcop}{\texttt{copy}}
\newcommand{\copc}[3]{\texttt{copy\_c}(#1, #2, #3)}
\newcommand{\pcopc}{\texttt{copy\_c}}
\newcommand{\copp}[2]{\texttt{copy\_p}(#1, #2)}
\newcommand{\pcopp}{\texttt{copy\_p}}
\newcommand{\sync}{\texttt{sync}}
\newcommand{\clock}[3]{\texttt{incUp}(#1, #2, #3)}
\newcommand{\pclock}{\texttt{incUp}}
\newcommand{\counter}[3]{\texttt{decUp}(#1, #2, #3)}
\newcommand{\pcounter}{\texttt{decUp}}
\newcommand{\shift}[1]{\texttt{shift}(#1)}
\newcommand{\pshift}{\texttt{shift}}
\newcommand{\update}[1]{\texttt{update}(#1)}
\newcommand{\pupdate}{\texttt{update}}
\newcommand{\erase}[1]{\texttt{erase}(#1)}
\newcommand{\perase}{\texttt{erase}}
\newcommand{\expand}[1]{\texttt{expand}(#1)}
\newcommand{\pexpand}{\texttt{expand}}

\newcommand{\size}[1]{\texttt{size}(#1)}

\newcommand{\seq}[1]{\texttt{sequence}(#1)}

\newcommand{\convp}[1]{\texttt{fix1}(#1)}
\newcommand{\pconvp}{\texttt{fix1}}
\newcommand{\convn}[1]{\texttt{fix0}(#1)}
\newcommand{\pconvn}{\texttt{fix0}}
\newcommand{\simp}[1]{\texttt{simp}(#1)}
\newcommand{\psimp}{\texttt{simp}}
\newcommand{\compa}[1]{\texttt{comp1}(#1)}
\newcommand{\pcompa}{\texttt{comp1}}
\newcommand{\compb}[1]{\texttt{comp2}(#1)}
\newcommand{\pcompb}{\texttt{comp2}}
\newcommand{\comp}[1]{\texttt{comp}(#1)}
\newcommand{\pcomp}{\texttt{comp}}
\newcommand{\allzero}{\ensuremath{(\zero^n,\zero^m)}}
\newcommand{\allone}{\ensuremath{(\one^n,\one^m)}}
\newcommand{\altzero}{\ensuremath{((\zero\one)^{\frac{n}{2}},(\zero\one)^{\frac{m}{2}})}}
\newcommand{\altone}{\ensuremath{((\one\zero)^{\frac{n}{2}},(\one\zero)^{\frac{m}{2}})}}

\newcommand{\instituteISSS}{I3S, UMR7271 CNRS et Universit{\'e} de Nice Sophia Antipolis, 06900 Sophia Antipolis, France}

\newcommand{\instituteIBISC}{IBISC, EA4526, Universit{\'e} d'{\'E}vry Val-d'Essonne, 91000 {\'E}vry, France}

\newcommand{\instituteLIF}{Aix-Marseille Universit{\'e}, CNRS, LIF UMR 7279, 13000 Marseille, France}

\newcommand{\instituteIXXI}{IXXI, Institut rh{\^o}ne-alpin des syst{\`e}mes complexes, 69000 Lyon, France}

\newcommand{\writeDTmail}{\{\texttt{tarek.melliti}, \texttt{damien.regnault}\}\texttt{@ibisc.univ-evry.fr}}

\newcommand{\writeJmail}{\texttt{jeremy.sobieraj}\texttt{@gmail.com}}

\newcommand{\writeMmail}{\texttt{mathilde.noual}\texttt{@unice.fr}}

\newcommand{\writeSmail}{\texttt{sylvain.sene}\texttt{@univ-amu.fr}}

\newcommand{\Mat}{Mathilde Noual}
\newcommand{\Jer}{J{\'e}r{\'e}my Sobieraj}
\newcommand{\Dam}{Damien Regnault}
\newcommand{\Tar}{Tarek Melliti}
\newcommand{\Syl}{Sylvain Sen{\'e}}

\newcommand{\Mshort}{M. Noual}

\newcommand{\Dshort}{D. Regnault}
\newcommand{\Tshort}{T. Melliti}
\newcommand{\Sshort}{S. Sen{\'e}}

\begin{document}

\bibliographystyle{splncs_srt}


\title{Full characterisation of attractors of two tangentially intersected asynchronous Boolean automata cycles}

\titlerunning{Asynchronous dynamics Boolean automata double-cycles}

\author{\Tar\inst{1} \and \Mat\inst{2} \and \Dam\inst{1} \and \Syl\inst{3,4} \and \Jer\inst{1}}

\authorrunning{\Tshort{}, \Mshort{}, \Dshort{}, \Sshort{} and \Jshort{}}

\institute{\instituteIBISC{}\\(\writeDTmail{}, \writeJmail{}) \and \instituteISSS{}\\(\writeMmail{}) \and \instituteLIF{}\\(\writeSmail{}) 
	\and \instituteIXXI}

\maketitle

\begin{abstract}
	The understanding of Boolean automata networks dynamics takes an important place in various domains of computer science such as 
	computability, complexity and discrete dynamical systems. In this paper, we make a step further in this understanding by focusing on 
	their cycles, whose necessity in networks is known as the brick of their complexity. We present new results that provide a 
	characterisation of the transient and asymptotic dynamics, \emph{i.e.} of the computational abilities, of asynchronous Boolean automata 
	networks composed of two cycles that intersect at one automaton, the so-called double-cycles. To do so, we introduce an efficient 
	formalism inspired by algorithms to define long sequences of updates, that allows a better description of their dynamics than previous 
	works in this area.
\end{abstract}


\section{Introduction}
\label{sec:intro}

Interaction networks occupying a ceaselessly increasing space in the knowledge of the objects that surround or even constitute 
us (\emph{e.g.} genetic regulation networks) as well as in our daily life (\emph{e.g.} social networks), it is now necessary to understand 
more their intrinsic properties. This paper follows this statement by using automata networks (ANs) as models of interaction networks. 
ANs have been chosen for two major reasons. First, although this computational model is among the firsts 
developed~\cite{McCulloch1943,vonNeumann1966}, lots of their intrinsic computational properties are not known nowadays. Second, their 
simplicity, and the concepts and parameters needed to define them, make them particularly adapted to capture the essence of, and model, real 
interaction systems at a high abstraction level, such as physical, biological and social systems~\cite{Ising1925,Kauffman1969,Schelling1971}.
The present work precisely takes place at the frontier of theoretical computer science and fundamental bio-informatics, that aims at 
analysing and explaining formally the dynamics of biological regulations, that have constituted the core of molecular 
biology~\cite{Jacob1961,Jacob1960}.\smallskip

Fundamental bio-infor\-matics gives rise to many theoretical and applied questions. In this context, Boolean automata networks (BANs) 
play a leading role. Indeed, since the seminal works of Kauffman~\cite{Kauffman1969,Kauffman1971} and Thomas~\cite{Thomas1973,Thomas1981} in 
theoretical biology, computer scientists have not stopped trying to answer their questions/conjectures. Among the latter, those that are 
central in this work are Thomas' ones, for which solutions have been proven in the discrete framework at the end of 
2000's~\cite{Remy2008,Richard2010,Richard2007}. These results, together with those of Robert~\cite{Robert1986}, highlighted that the ability 
of ANs to admit complex asymptotic behaviours only comes from the presence of cycles in their architecture. However, although the fact that 
interacting cycles are the engines of dynamical complexity is known, we don't really/perfectly know how yet. That explains why many recent 
studies focused on these specific patterns. Among them,~\cite{Demongeot2012} gave the characterisation in parallel of the dynamical 
behaviours of Boolean automata cycles (BACs). Then, time was attached to analyse the relations between the dynamical properties of cycles 
subjected to distinct updating modes, with a special attention paid to the asynchronous and the parallel ones~\cite{Noual2012}. Once the 
cycle dynamics finely understood, the natural idea was to study more complex networks. But to obtain general results for any kind of network 
remains an open problem that seems intractable at present. So, following a constructive approach and as a first step, studies have been led 
on specific patterns combining cycles, such as the double-cycles in parallel~\cite{Demongeot2012} and the flower-graphs~\cite{Didier2012} for 
instance. In addition, other studies have dealt with the convergence time of specific classes of BANs, like circular \textsc{xor} 
networks~\cite{Noual2013} and networks without negative cycles~\cite{Melliti2013}.\smallskip

This paper follows the same lines and solves a question that remained open until now: how do Boolean automata double-cycles that evolve 
asynchronously over time behave? The answer is given by emphasising original methods for the domain in the sense that they are very 
algorithmic. In particular, they allow to show that recurrent configurations are not all similar (some have peculiar features). Some of them 
can be reached by following paths of linear size according to the network sizes whereas other need quadratic sequences of updates to be 
reached. In fact, the results presented give a deeper characterisation of the attractors.\medskip

The paper is organised as follows: Section~\ref{sec:def} gives the main definitions and notations used in the paper, in particular those 
related to the double-cycles and the asynchronous updating mode; Section~\ref{sec:tools} gives the definition of the tools and methods 
developed here; finally, Section~\ref{sec:res} is dedicated to the main contributions of this paper.

\section{Definitions and notations}
\label{sec:def}

\paragraph{BANs.~} Consider $\BB = \{\zero,\one\}$ and $\text{V} = \{0, \ldots, n-1\}$ a set of $n$ Boolean automata so that $\forall i \in 
\text{V},\ x_i \in \BB$ denotes the \emph{state} of $i$. A \emph{configuration} of a BAN $\mathscr{N}$ of size $n$ instantiates the state of 
any $i$ of $V$ and is classically denoted as a vector, such that $x \in \BB^n$, or as a binary word. Formally, a BAN $\mathscr{N}$, whose 
automata set is $\text{V}$, is a set of $n$ Boolean functions, which means that $\mathscr{N} = \{\text{f}_i: \BB^n \to \BB\ |\ i \in V\}$. 
Given $i \in \text{V}$, $\text{f}_i$ is the \emph{local transition function} of $i$ that predetermines its evolution for any configuration 
$x$. Actually, that means that if $i$ is updated in $x$, its state switches from $x_i$ to $\text{f}_i(x)$. Let us define now the \emph{sign 
of an interaction} from $j$ to $i$ ($i, j \in \text{V}$) in configuration $x \in \BB^n$ with $\text{sign}_x(j,i) = \text{s}(x_j) \cdot 
(\text{f}_i(x) - \text{f}_i(\bar{x}^j))$, where $\text{s}: \BB \to \un$, with $\un = \{-1, 1\}$, is defined as $\text{s}(b) = b - \neg b$, 
and $\forall i \in \text{V},\ \bar{x}^i = (x_0, \ldots, x_{i-1}, \neg x_i, x_{i+1}, \ldots, x_{n-1})$. Interactions that are \emph{effective} 
in $x$ belongs to the set $\text{A}(x) = \{(j, i) \in \text{V}^2\ |\ \text{sign}_x(j, i) \neq 0\}$. From this is derived the 
\emph{interaction graph} of $\mathscr{N}$ that is the digraph $\text{G} = (\text{V}, \text{A})$, where $\text{A} = \bigcup_{x \in \BB^n} 
\text{A}(x)$ is the set of interactions.\smallskip

In this paper, the focus is put on BANs associated with \emph{simple} interaction graphs: if there exists $(j,i) \in \text{A}$, it is unique 
and such that $\forall x \in \BB^n, \text{sign}_x(j,i) \neq 0$ and is constant. As a consequence, $\text{sign}(j,i) \in \un$. If $\text{sign}
(j,i) = +1$ (resp. $\text{sign}(j,i) = -1$), $(j,i)$ is an activating (resp. inhibiting) interaction so that the state of $i$ tends to mimic 
(resp. negate) that of $j$. We call the \emph{signed interaction graph} of $\mathscr{N}$ the digraph obtained by labelling each arc $(i,j) 
\in \text{A}$ with $\text{sign}(i,j)$. In order not to burden the reading, we also denote it by $\text{G}$. Abusing notations, a \emph{cycle} 
$\text{C}$ of $\text{G}$ is said to be \emph{positive} (resp. \emph{negative}) if the product of the signs of the interactions that compose 
it equals $+1$ (resp. $-1$).

\paragraph{Asynchronous transition graphs.~} In a BAN $\mathscr{N}$, a couple of configurations $(x,y) \in \BB^n \times \BB^n$, such that $y$ 
is obtained by updating the state of a unique automaton of $x$ is an \emph{asynchronous transition}, and is denoted by $x \trans y$ (the 
Hamming distance $\text{d}(x, y) \leq 1$). If $x \neq y$, $x \trans y$ is said to be \emph{effective}.
Let $\text{T} = \{x \trans y\ |\ x, y \in \BB^n\}$ be the set of asynchronous transitions of $\mathscr{N}$. Digraph $\mathscr{G} = (\BB^n, 
\text{T})$ is then the \emph{asynchronous transition graph} (abbreviated simply by \emph{transition graph}) of $\mathscr{N}$, which actually 
represents the \emph{non-deterministic ``perfectly'' asynchronous discrete dynamical system} related to $\mathscr{N}$.\smallskip

Consider an arbitrary BAN $\mathscr{N}$, its transition graph $\mathscr{G} = (\BB^n, T)$ and $x \in \BB^n$ any of its possible 
configurations. A \emph{trajectory} of $x$ is any path in $\mathscr{G}$ that starts in $x$. A strongly connected component of $\mathscr{G}$ 
that admits no outgoing asynchronous transitions is a \emph{terminal strongly connected component} (TSCC). A TSCC of $\mathscr{G}$ represents 
an asymptotic behaviour of $\mathscr{N}$, \emph{i.e.} one of its \emph{attractors}. A configuration that belongs to an attractor is a 
\emph{recurrent configuration} and, for a given attractor, the number of its configurations is said to be its \emph{size}. An attractor of 
size $1$ (resp. of size greater than $1$) is a \emph{stable configuration} (resp. \emph{a stable oscillation}). We close this paragraph by 
defining the \emph{convergence time of a configuration} $x$ as the length of the shortest trajectory that leads it to an attractor and the 
\emph{convergence time of a BAN} as the highest convergence time of all configurations in $\BB^n$.

\paragraph{Boolean automata double-cycles.~} The literature has put the emphasis on BACs. The reason comes from the three following theorems 
that show that cycles are necessary for BANs to admit complex asymptotic dynamics. Now, consider $\mathscr{G}$ as the asynchronous transition 
graph of a BAN $\mathscr{N}$.

\begin{figure}[t!]
	\centerline{\scalebox{0.65}{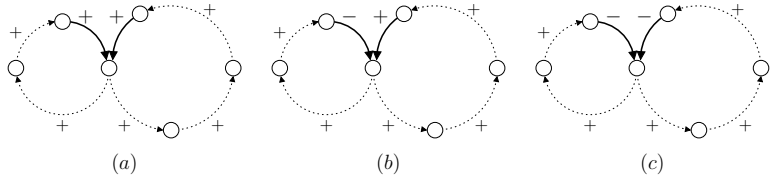}}
	\caption{Interaction graphs of the three kinds of canonical BADCs: $(a)$ a canonical positive BADC, $(b)$ a canonical mixed BADC, and 
		$(c)$ a canonical negative BADC.}
	\label{fig:badc}
\end{figure} 

\begin{theorem}
	\emph{\cite{Robert1986}}~Whatever the updating mode is, if $\mathscr{N}$ does not contain any cycle, then it admits a unique attractor, 
	that is a stable configuration.
\end{theorem}

\begin{theorem}
	\emph{\cite{Remy2008,Richard2007,Thomas1981}} If $\mathscr{G}$ admits two stable configurations then the interaction graph of 
	$\mathscr{N}$ contains a positive cycle.
\end{theorem}

\begin{theorem}
	\emph{\cite{Remy2008,Richard2010,Thomas1981}} If $\mathscr{G}$ admits a stable oscillation then the interaction graph of $\mathscr{N}$ 
	contains a negative cycle.
\end{theorem}

On the basis of the theorems above, and in the same lines as~\cite{Demongeot2012,Noual2012b} that characterises the dynamical behaviour in 
parallel of Boolean automata double-cycles (BADCs), we propose in this paper to study BADCs when updated asynchronously. Informally, a 
\emph{BADC} $\mathscr{D}$ of size $n + m - 1$ is composed of two BACs $\Cone$ (of size $n$) and $\Ctwo$ (of size $m$) that 
intersect tangentially at one automaton that will be denoted specifically, for the sake of clarity in proofs, by $c$ (resp. $\cone_0$, 
$\ctwo_0$) when considering $\mathscr{D}$ (resp. $\Cone$, $\Ctwo$). Notice that in $\mathscr{D}$, every automaton admits a unary function as 
its local transition function that is either $\textsf{id}$ or $\textsf{neg}$, except automaton $c$ that admits a binary function. In this 
paper, we focus on monotone functions and enforce $f_c$ to be the \textsc{and}-function without loss of generality for our concern. Also, 
remark that there exist three different kinds of BADCs: \emph{positive BADCs} made of two positive BACs, \emph{negative BADCs} made of two 
negative BACs, and \emph{mixed BADCs} made of one positive and one negative cycles. An interesting point is that the study of BADCs of size 
$n+m-1$ in general can be reduced to that of three \emph{canonical} BADCs of size $n+m-1$~\cite{Noual2012b,Noual2012}, presented in 
Figure~\ref{fig:badc}, because of the isomorphism between their transition graphs. A canonical positive BADC $\Dp$ is composed only of 
positive interactions. A canonical negative BADC $\Dm$ is composed only of positive interactions, except the two that have $c$ as their 
destination. A canonical mixed BADC $\Dpm$ is composed only of positive interactions, except one of those that have $c$ as their destination 
(we suppose that this interaction belongs to $\Cone$). To finish, for easing the proofs, we denote a BADC configuration $x$ by a vector of 
two binary words, in which the first symbol represents $x_c$. For instance, the null configuration in which all automata are at state $\zero$ 
is denoted by $(\zero^n, \zero^m)$. Also, we denote by $x^\ell$ (resp. $x^r$) the projection of $x$ on cycle $\Cone$ (resp. $\Ctwo$). Thus, 
$x = (x^\ell, x^r)$ and the state of automaton $c^\ell_i$ in configuration $x$ is $x^\ell_i$. Note that $x_0 = x^\ell_0 = x^r_0$ since these 
three notations stand for the state of automaton $\com$ in configuration $x$.

\section{Algorithmic tools}
\label{sec:tools}

In this section, we introduce the tools that will be used further to study the dynamics of BADCs. We introduce first the 
\emph{expressiveness} of a configuration, which counts the number of its $\zero\one$ patterns. This notion is inspired by works on 
asynchronous cellular automata that have shown that the occurrence number of this pattern is crucial to understand their 
behaviour~\cite{Fates2006}. Then are introduced \emph{instructions} to represent sequences of updates as classical algorithms. Instructions 
are used to express long sequences of updates with few lines of code.

\subsection{Expressiveness}

\begin{definition}
\label{def:expressiveness}
Let $x$ be a configuration of a BAC $\mathscr{C}$ of size $n$. The \emph{expressiveness} of $x$ is the number of $\zero\one$ patterns in $x$, 
\emph{i.e.} $|\{i\ |\ 0 \leq i \leq n-1, x_i = 0 \text{ and } x_{i+1 \mod n} = 1\}|$. 
\end{definition}

From Definition~\ref{def:expressiveness}, we derive easily the expressiveness of a configuration $x$ of a BADC $\mathscr{D}$ as the sum of 
the expressivenesses of $x^\ell$ and $x^r$. Expressiveness is very useful to understand the structure of attractors. The least expressive 
configurations are $\allzero$ and $\allone$ and the most expressive ones are $\altone$ and $\altzero$ (if $n$ and $m$ are even). In the 
sequel, we will see that: \emph{(i)} the lowly expressive configurations generally are recurrent and can be reached in linear time 
by most configurations; \emph{(ii)} the highly expressive configurations either are not recurrent or can only be reached through very 
specific update sequences, and they can quickly reach any other configuration. So, for a BADC $\mathscr{D}$ that admits an attractor of 
exponential size made of lowly expressive and highly expressive configurations, we conjecture that: $(1)$ \emph{the shortest path from a 
highly expressive configuration to any other configuration is linear in $n$ and $m$}; $(2)$ \emph{the shortest path from a lowly expressive 
configuration to a highly expressive one is quadratic in $n$ and $m$}. In other terms, to decrease expressiveness is easy whereas to increase
expressiveness is hard.

\subsection{Elementary instructions}

In this article, lots of proofs rely on exhibiting update sequences between two configurations. However, the length of such sequences is 
problematic and a human reader would not manage to extract directly from these sequences the proof general ideas. Thus, we propose to view 
update sequences as \emph{instructions} that allow to define them and understand their effect on configurations easily. 

Let $\mathscr{D}$ be a BADC, $\mathscr{C}$ be one of the BACs of $\mathscr{D}$, $x$ the current configuration of $\mathscr{C}$, and $c_i$ and 
$c_j$ be two automata of $\mathscr{C}$ distinct of $\com$ and such that $i < j$. In the sequel, we will make particular use of the following 
elementary instructions:

\begin{itemize}
\item[$\bullet$] $\sync$: $x_c \leftarrow f_c(x)$ \hfill \# \textsf{\small update of $\com$}\\
	$\sync$ is the only instruction that updates automaton $\com$ and where both BACs interact with each other. This (key)-instruction will 
	always be called when $\com$ can change its state. $\sync$ can be used either to set $\com$ at a desired state or to increase the 
	expressiveness from a configuration. Notice that $\sync$ is the only way to switch a $\one\one\one$ (resp. $\zero\zero\zero$) pattern
	into a $\one\zero\one$ (resp. $\zero\one\zero$) pattern and, thus, to increase the expressiveness. Remark that the BAC
	sub-configurations have to be specific for $\com$ to switch its state.\\[-3mm]
\item[$\bullet$] $\update {c_i}$: $x_{c_i} \leftarrow f_{c_i}(x)$ \hfill \# \textsf{\small update of $c_i$}\\
	$\pupdate$ updates an automaton distinct to $\com$. \\[-3mm]
\item[$\bullet$] $\clock {\mathscr{C}} {i} {j}$: \texttt{for} $k = i$ \texttt{to} $j$ \texttt{do} $\update {c_k}$ \hfill \# \textsf{\small 
	incremental updates}\\
	$\pclock$ updates consecutive automata by increasing order. In fact, $\pclock$ aims at propagating the state of $c_{i-1}$ along 
	$\mathscr{C}$. Notice that if $j<i$ then no automata are updated. Moreover, since $i \neq 0$, $\com$ cannot be updated with $\pclock$.
	\vspace*{-1mm}\begin{property}
	\label{prop:clock}
		Let $x'$ be the result of applying $\clock {\mathscr{C}} {i} {j}$ on configuration $x$. Then we have: $\forall k \in 
		\{i, \ldots, j\},\ x'_k = x_{i-1}$ and $\forall k \notin \{i,\ldots, j\},\ x'_k = x_k$.
	\end{property}\vspace*{-1mm}
\item[$\bullet$] $\erase {\mathscr{C}}$: $\clock {\mathscr{C}} {1} {\size {\mathscr{C}} -1}$\\
	$\perase$ is a particular case of $\pclock$ that aims at propagating the state of $c_0$ along $\mathscr{C}$. As a consequence, using this 
	instruction on $\mathscr{C}$ makes it to be of expressiveness $0$, and thus, is really efficient to converge quickly to a stable 
	configuration of least expressiveness (if should be the case).  
	\vspace*{-1mm}\begin{property}
	\label{prop:erase}
		Let $x'$ be the result of applying $\erase {\mathscr{C}}$ on configuration $x$. Then we have: $\forall k \in \{0, \ldots, \size 
		{\mathscr{C}} -1\},\ x'_k = x_{0}$.
	\end{property}\vspace*{-1mm}
\item[$\bullet$] $\expand {\mathscr{C}}$: $\clock {\mathscr{C}} {1} {\kappa -1 \in \mathbb{N}}$ with\\ 
	$\kappa = \underset{1 \leq k \leq \size {\mathscr{C} -1}}{\min} \left\lbrace k \ |\ \begin{cases}
		(x_k = 0) \text{ and } (x_{k+1 \mod \size {\mathscr{C}}} = 1) & \text{if } x_c = 1\\
		(x_k = 1) \text{ and } (x_{k+1 \mod \size {\mathscr{C}}} = 0) & \text{if } x_c = 0
	\end{cases} \right\rbrace$.\\
	$\pexpand$ is another particular case of $\pclock$ that aims at propagating the state of $c_0$ along $\mathscr{C}$ while neither 
	$\zero\one$ nor $\one\zero$ patterns are destroyed, which avoids decreasing the expressiveness of $\mathscr{C}$. 
\item[$\bullet$] $\counter {\mathscr{C}} {i} {j}$: \texttt{for} $k = j$ \texttt{down to} $i$ \texttt{do} $\update {c_k}$ \hfill \# 
	\textsf{\small decremental updates}\\
	$\pcounter$ updates consecutive automata by decreasing order. Once $\counter {\mathscr{C}} {i} {j}$ executed, the information of $c_j$ is 
	lost and that of $c_{i-1}$ is possessed by both $c_{i-1}$ and $c_i$. In fact, $\pcounter$ aims at shifting partially a BAC section. As 
	for $\pclock$, if $j<i$ then no automata are updated and $\com$ cannot be updated with $\pcounter$.
	\vspace*{-1mm}\begin{property}
	\label{prop:counter}	
		Let $x'$ be the result of applying $\counter {\mathscr{C}} {i} {j}$ on configuration $x$. Then we have: $\forall k \in \{i, \ldots, 
		j\},\ x'_k = x_{k-1}$ and $\forall k \notin \{i, \ldots, j\},\ x'_k = x_k$.
	\end{property}
\item[$\bullet$] $\shift {\mathscr{C}}$: $\counter {\mathscr{C}} 1 {\size {\mathscr{C}} -1}$\\
	$\pshift$ is a particular case of $\pcounter$. Once executed, every automaton of $\mathscr{C}$ takes the state of its predecessor, except 
	$\com$ whose state does not change. Automaton $c_{\size {\mathscr{C} -1}}$ excluded, all the information contained along $\mathscr{C}$ is 
	kept safe. To use $\pshift$ is useful to propagate information along a BAC without loosing too much expressiveness (at most one 
	$\zero\one$ pattern is destroyed).
\end{itemize}

\section{Results}
\label{sec:res}

\subsection{More complex instructions}
\label{sec:comp_inst}

Now, consider a configuration $x$ of BADC $\mathscr{D}$ and an algorithm made of instructions that defines a sequence of updates (abbreviated 
simply by ``sequence'' from now) from $x$, denoted by $\seq{x}$. Abusing language, in the sequel, $\seq{x}$ represents both the underlying 
sequence and its result, namely the configuration resulting from the execution of $\seq{x}$. To end this section, we introduce three other 
sequences in Table~\ref{tab:algo0}, more complex, that will be important later. In particular, Lemma~\ref{lem:copy} states that $\pcop$ 
allows to transform $x$ into $x'$ if $x$ is expressive enough (highly expressive actually). 

\begin{table}[t!]
	\caption{The sequences $\pcopc$, $\pcop$ and $\pcopp$.}
	\label{tab:algo0}
	\scalebox{.9}{\hspace*{2mm}\centerline{
		\begin{minipage}{.52\textwidth}
			\centerline{
				\begin{tabular}{m{\textwidth}}
					\fbox{$\copc {x} {x'} {\mathscr{C}^s}$}\\
					01.~ $n \leftarrow \size {\mathscr{C}^s}$;\\
					02.~ \textbf{if} ($x^s_{n-1} = x^s_{n-2} \text{ and } x^s_{n-1} \neq x'^s_{n-1}$) \textbf{then}\\
					03.~ ~~~$j \leftarrow \max \{k\ |\ k < n-1 \text{ and } x^s_k \neq x'^s_k\}$;\\
					04.~ \textbf{else} $j \leftarrow n$;\\
					05.~ \textbf{end if}\\
					06.~ \textbf{for} ($k = n-1$) \textbf{down to} ($j + 1$) \textbf{do}\\
					07.~ ~~~$\update {c^s_{k-1}}$;\\
					08.~ ~~~$\update{c^s_{k}}$;\\
					09.~ \textbf{done}\\
					10.~ \textbf{for} ($k = j-1$) \textbf{down to} ($1$) \textbf{do}\\
					11.~ ~~~\textbf{if} ($x^s_k \neq x'^s_k$) \textbf{then} $\update{c^s_k}$;\\
					12.~ ~~~\textbf{end if}\\
					13.~ \textbf{done}
				\end{tabular}
			}
		\end{minipage}
		\hspace*{5mm}
		\begin{minipage}{.23\textwidth}
			\centerline{
				\begin{tabular}{m{\textwidth}}
					\fbox{$\cop {x} {x'}$}\\
					01.~ $\copc {x} {x'} {\Cone}$;\\
					02.~ $\copc {x} {x'} {\Ctwo}$;
				\end{tabular}
			}
			\vspace*{5mm}
			\centerline{
				\begin{tabular}{m{\textwidth}}
					\fbox{$\copp {x} {x'}$}\\
					01.~ \textbf{if} ($\com \neq \com'$) \textbf{then}\\
					02.~ ~~~$\shift{\Cone}$;\\
					03.~ ~~~$\shift{\Ctwo}$;\\
					04.~ ~~~$\sync$;\\
					05.~ \textbf{end if}\\
					06.~ $\cop {x} {x'}$;
				\end{tabular}
			}
		\end{minipage}
	}}
\end{table}	

\begin{lemma}
	\label{lem:copy}
	Let $\mathscr{D}$ be a BADC and $x$ and $x'$ two of its configurations such that $x_0 = x'_0$. If, for any $s \in \{\ell, r\}$, one of 
	the following properties holds for $x$:
	\begin{enumerate}
	\item $\forall i \in \{1, \ldots, \size{\mathscr{C}^s} -1\},\ x^s_i \neq x^s_{i-1}$,
	\item $\forall i \in \{1, \ldots, \size{\mathscr{C}^s} -2\},\ x^s_i \neq x^s_{i-1} \text{ and } x^s_{\size{\mathscr{C}^s} -1} = 
		x'^s_{\size{\mathscr{C}^s} -1}$,
	\item $\forall i \in \{1, \ldots, \size{\mathscr{C}^s} -2\},\ x^s_i \neq x^s_{i-1} \text{ and } \exists p \in \{1, \ldots, 
		\size{\mathscr{C}^s} -2\}, x^s_p \neq x'^s_p$,
	\end{enumerate}
	then $\cop {x} {x'} = x'$ and this sequence consists in at most $2(n + m - 6)$ updates.
\end{lemma}

\begin{proof}
	Remark that $\sync$ is never called in $\pcop$. Thus, the state of $c$ never switches and $x_0 = x'_0$. Since $\pcop$ calls twice 
	$\pcopc$, once on $\Cone$ and then on $\Ctwo$, let us focus without loss of generality on $\copc {x} {x'} {\Cone}$ and prove that this 
	sequence transforms $x^\ell$ into $x'^\ell$ (the same kind of reasoning adapts directly to $\copc {x} {x'} {\Ctwo}$).
	
	First, it is important to notice that, if $x^\ell$ follows either Property~1 or Property~2, which both induce that the value of $j$ is 
	initialised to $n$, the only $\texttt{for}$-loop that can be executed is that of line 10. Now, the assumption stating that $\forall i \in 
	\{1, \ldots, \size{\mathscr{C}^s} -1\},\ x^s_i \neq x^s_{i-1}$ together with lines 11-13 make $x^\ell$ to become $x'^\ell$.
	
	Second, let us focus on a configuration $x^\ell$ for which Property~3 holds but not Properties~1 and~2. Such an $x^\ell$ necessarily 
	verifies conditions given in line 2, which leads $j$ to be well defined since, by hypothesis, $\exists p \in \{1, \ldots, \size{\Cone} 
	-2\}, x^\ell_p \neq x'^\ell_p$ (notice that $j$ is set to the greatest $p$ satisfying this relation). As a consequence, the content of 
	the $\texttt{for}$-loop of line 6 is executed. Let us now prove that, at the end of the execution of this loop, $\forall j < k < n-1,\ 
	x^\ell_k = x'^\ell_k$. From this, consider the following loop invariant \emph{\bf inv($k$)}: ``at the beginning of the $k$-th iteration, 
	$x^\ell_{k-1} = x^\ell_k$ and $x^\ell_{k-1} \neq x^\ell_{k-2}$.''
	
	For the $(n-1)$-th iteration, from above, the invariant holds.

	Assume that the invariant still holds at the $k$-th iteration. Given that $x^\ell_{k-1} \neq x^\ell_{k-2}$, line 7 makes $c^\ell_{k-1}$ 
	switch its state that consequently becomes \emph{(i)} different from that of $c^\ell_k$ and \emph{(ii)} equal to that of $c^\ell_{k-2}$. 
	Then, because of \emph{(i)}, line 8 makes $c^\ell_k$ switch. Notice that, at this point, the states of $c^\ell_{k-2}$ and $c^\ell_{k-3}$ 
	have not been changed and $x^\ell_{k-2} \neq x^\ell_{k-3}$. Thus, with \emph{(ii)}, the invariant still holds for the $(k-1)$-th 
	iteration.
	
	According to what has just been explained, at the end of the loop, every automaton $c^\ell_k$, $j < k < n-1$ has switched twice (and thus 
	has recovered its initial state) whereas automata $c^\ell_j$ and $c^\ell_{n-1}$ have switched once (and thus do have changed their 
	state). As a consequence, we now have that $x^\ell_{n-1} = x^\ell_{n-1}$ and $x^\ell_j = x'^\ell_j$. All this ensures that at line 9, 
	$\forall j \leq k \leq n-1,\ x^\ell_k = x'^\ell_k$.
	
	For ending the proof, with $0 \leq k \leq j-1$, it suffices to follow the $\texttt{for}$-loop of line 10 whose effect has been explained 
	in the previous paragraph. Also, we have just seen that in $\Cone$, $\pcopc$ can lead $(n-2)$ automata (except $c_j$ and $c_{n-1}$ as 
	said before) to switch twice in the worst case, \emph{i.e.} when $j$ is initialised to $1$.\smallskip 

	As a consequence, the execution of $\pcop$ takes at most $2(n-2)-2 + 2(m-2)-2 = 2(n + m - 6)$ updates.\qed
\end{proof}

From this first result that gives strong insights about the power of instructions and sequences to reveal possible trajectories between configurations, let us now focus on the dynamical behaviours of double-cycles.

\subsection{Positive BADCs}
\label{sec:pbadc}

Since results of~\cite{Noual2012b,Noual2012} have shown that positive BADCs behave as positive BACs, and because stable configurations 
are conserved between distinct updating modes~\cite{Goles1990}, it is easy to show that the asymptotic dynamics of positive BADCs consists in 
two stable configurations $x$ and $\bar{x}$ (where $\bar{x}$ denotes the negation of $x$). In the case of canonical BADCs, these stable 
configurations are $\allzero$ and $\allone$. Here, let us focus on an arbitrary positive BADC $\Dp$. We show that two new sequences $\pconvn$ 
and $\pconvp$ (cf. Table~\ref{tab:algo1}) can respectively transform any configuration with at least one automaton at state $\zero$ into 
$\allzero$, and any configuration with at least one automaton at state $\one$ in both cycles into $\allone$.

\begin{table}[t!]
	\caption{The sequences $\pconvn$ and $\pconvp$.}
	\label{tab:algo1}
	\scalebox{.9}{\hspace*{6mm}\centerline{
		\begin{minipage}{.34\textwidth}
			\centerline{
				\begin{tabular}{m{\textwidth}}
					\fbox{$\convn {x}$}\\
					01.~ \textbf{if} ($x_0$ = 1) \textbf{then}\\
					02.~ ~~~$i \leftarrow \min \{k\ |\ x^\ell_k = 0\}$;\\
					03.~ ~~~$\clock {\Cone} {i+1} {n-1}$;\\
					04.~ ~~~$\sync$;\\
					05.~ \textbf{end if}\\
					06.~ $\erase {\Cone}$;\\
					07.~ $\erase {\Ctwo}$;
				\end{tabular}
			}
		\end{minipage}
		\hspace*{5mm}
		\begin{minipage}{.35\textwidth}
			\centerline{
				\begin{tabular}{m{\textwidth}}
					\fbox{$\convp {x}$}\\
					01.~ \textbf{if} ($x_0$ = 0) \textbf{then}\\
					02.~ ~~~$i \leftarrow \min \{k\ |\ x^\ell_k = 1\}$;\\
					03.~ ~~~$\clock {\Cone} {i+1} {n-1}$;\\
					04.~ ~~~$j \leftarrow \min \{k\ |\ x^r_k = 1\}$;\\
					05.~ ~~~$\clock {\Ctwo} {j+1} {m-1}$;\\
					06.~ ~~~$\sync$;\\
					07.~ \textbf{end if}\\
					08.~ $\erase {\Cone}$;\\
					09.~ $\erase {\Ctwo}$;
				\end{tabular}
			}
		\end{minipage}
	}}\vspace*{-3mm}
\end{table}

\begin{theorem}
	\label{thm:pbadc}
	Let $\Dp$ be a canonical positive BADC and $x$ one of its unstable configuration. If $x$ admits one automaton at state $0$, then $\convn 
	{x} = \allzero$. Also, if $x$ admits one automaton at state $1$ in both its cycles, then $\convp {x} = \allone$. The convergence time of 
	$\Dp$ is at most $2(n + m) - 5$. 
\end{theorem}

\begin{proof}
	Let us focus on the case of $\pconvp$ and consider an unstable configuration $x$ of $\Dp$ with at least one $\one$ in both cycles. First, 
	if $c$ is at state $1$, $\perase$ of lines 8 and 9 make every automaton of $\Cone$ and $\Ctwo$ to take state $\one$ and the obtained 
	configuration is then $\allone$ which is stable.
	
	Second, consider that $c$ is at state $0$. So, instructions of lines 2-6 are executed. Since, by hypothesis, there is at least a $\one$ 
	in $\Cone$ and $\Ctwo$, after the execution of incUp at line 3, $\forall k \in \{i, \ldots, n-1\},\ x^\ell_k = 1$, and, after the 
	execution of incUp at line 5, $\forall k \in \{j, \ldots, m-1\},\ x^r_k = 1$. As a consequence, the effect of $\sync$ at line 6 is to fix 
	$c$ at state $\one$ and we get back to the case above.
	
	Now, notice that the case of $\pconvn$ is very similar, by considering with no loss of generality that at least one 
	automaton is at state $0$ in $\Cone$ and that we need to set $x^\ell_{n-1}$ to $\zero$ before the execution of $\sync$ at line 
	4.\smallskip
	
	Finally, notice that the number of effective updates made by $\pconvn$ (resp. $\pconvp$) is at most $2n + m - 3$ (resp. $ = 2(n + m) - 
	5$).\qed 
\end{proof}

\subsection{Mixed BADCs}
\label{sec:mbadc}

Now, we pay attention to mixed BADCs. From the same works that showed also that asynchronism keeps only recurrent configurations of least 
global instability, we know that their asymptotic dynamics consists only in a stable configuration. In particular, the attractor of canonical 
mixed BADCs is $\allzero$. Let us focus on their convergence time. To do so, we will make particular use of new sequence $\psimp$ (cf. 
Table~\ref{tab:algo2}) that gives a way of converging to this stable configuration from any initial configuration $x$, by reducing 
progressively its expressiveness.

\begin{theorem}
	\label{thm:mbadc}
	Let $\Dpm$ be a canonical mixed BADC. For any configuration $x$ of $\Dpm$, $\simp {x} = \allzero$ holds. The convergence time of $\Dpm$ 
	is at most $2n + m - 2$. 
\end{theorem}

\begin{proof}
	First, if $c$ is at state $0$, $\perase$ of lines 5 and 6 make every automaton of $\Cone$ and $\Ctwo$ to take state $\zero$ and the 
	stable configuration $\allzero$ is obtained.
	
	Second, consider that $c$ is at state $1$. Instructions of lines 2 and 3 are thus executed. So, $\perase$ makes every automaton of 
	$\Cone$ take state $1$, and $\sync$ makes $c$ take state $0$. And we get back to the case above.\smallskip
	
	Finally, notice that the number of effective updates made by $\psimp$ is at most $2n + m - 2$.\qed 
\end{proof}

\subsection{Negative BADCs}
\label{sec:nbadc}

In this section, we interest in negative BADCs. Contrary to BADCs of other sorts, the previous results 
of~\cite{Demongeot2012,Noual2012b,Noual2012} obtained under the parallel updating mode are not helpful for dealing with the 
asynchronous updating mode. Indeed, in parallel, negative BADCs admit an exponential number of attractors. In our asynchronous framework, we 
will show that they admit a unique stable oscillation of exponential size that depends on the parity of underlying cycles. In particular, the 
study that follows is divided in two axes: the first one deals with BADCs made of two negative cycles of even sizes (abbreviated by $\Dm_e$), 
the second one with the others where at least one cycle of odd size (abbreviated by $\Dm_o$).

\begin{table}[t!]
	\caption{The sequences $\psimp$, $\pcompa$ and $\pcompa$.}
	\label{tab:algo2}
	\scalebox{.9}{\hspace*{16mm}\centerline{
		\begin{minipage}{.24\textwidth}
			\centerline{
				\begin{tabular}{m{\textwidth}}
				\fbox{$\simp {x}$}\\
				01.~ \textbf{if} ($x_0$ = 1) \textbf{then}\\
				02.~ ~~~$\erase {\Cone}$;\\
				03.~ ~~~$\sync$;\\
				04.~ \textbf{end if}\\
				05.~ $\erase {\Cone}$;\\
				06.~ $\erase {\Ctwo}$;
				\end{tabular}
			}
		\end{minipage}
		\hspace*{5mm}
		\begin{minipage}{.33\textwidth}
			\centerline{
				\begin{tabular}{m{\textwidth}}
				\fbox{$\compa {x}$}\\
				01.~ \textbf{for} ($i = 1$) \textbf{to} ($n-1$) \textbf{do}\\
				02.~ ~~~ $\sync$;\\
				03.~ ~~~ $\expand {\Cone}$;\\
				04.~ ~~~ $\erase {\Ctwo}$;\\
				05.~ \textbf{done}
				\end{tabular}
			}
		\end{minipage}
		\hspace*{5mm}
		\begin{minipage}{.52\textwidth}
			\centerline{
				\begin{tabular}{m{\textwidth}}
				\fbox{$\compb {x}$}\\
				01.~ \textbf{if} ($x^r= \one^m$) \textbf{then}\\
				02.~ ~~~$\sync$;\\
				03.~ ~~~$\erase {\Ctwo}$;\\
				04.~ \textbf{end if}\\
				04.~ $\sync$;\\
				06.~ $\expand {\Ctwo}$;\\
				07.~ \textbf{for} ($i = 1$) \textbf{to} ($m-2$) \textbf{do}\\
				08.~ ~~~$\shift {\Cone}$;\\
				09.~ ~~~$\sync$;\\
				10.~ ~~~$\expand {\Ctwo}$;\\
				11.~ \textbf{done}
				\end{tabular}
			}
		\end{minipage}
	}}\vspace*{-3mm}
\end{table}	

\vspace*{-3mm}\subsubsection{Both cycles are even}

Here, we show that any BADC $\Dm_e$ admits only one stable oscillation of size $2^{n+m-1}$. In other terms, all configurations are recurrent 
and the convergence time is null. However, although all configurations are accessible from each other, those of high expressiveness are hard 
to reach. The proof of this result follows three points (they will be referred to Points~1,~2 and~3 later) in which it is respectively shown 
that:

\begin{enumerate}
\item any configuration can reach the least expressive one $\allzero$ in linear time;
\item configuration $\allzero$ can reach the highest expressive one $\altone$ in quadratic time;
\item any configuration can be reached from $\altone$ in linear time.
\end{enumerate}

Notice that Point~2 above is the hardest part. Indeed, to reach $\altone$ from $\allzero$ needs $O(n^2 + m^2)$ updates.
We will see that this upper bound is tight and that to increase a configuration expressiveness by $\delta$ requires at least $\delta^2$ 
updates (cf. Theorem~\ref{thm:quad}).\medskip

Let us consider Point~1. It is easy to see that sequence $\psimp$ is still efficient to reach $\allzero$ and thus, that the following Lemma 
holds.

\begin{lemma}
	\label{lem:simply}
	For any configuration $x$ of $\Dm_e$, $\simp {x} = \allzero$ holds and takes at most $2n + m - 2$ updates. 
\end{lemma}

\begin{proof}
	This proof is identical to that of Theorem~\ref{thm:mbadc}, except the fact that $\allzero$ is not a stable configuration anymore.\qed
\end{proof}

Now, let us pay attention to Point~2 that asks for increasing the expressiveness of $\allzero$. We characterise here a path from this 
configuration to $\altone$. To do so, let us proceed in two steps. The first one aims at increasing the expressiveness of $\Cone$ by means of 
sequence $\pcompa$ (cf. Lemma~\ref{lem:complex1}), the second at increasing that of $\Ctwo$ while ensuring not to decrease that of $\Cone$ by 
means of $\pcompb$ (cf. Lemma~\ref{lem:complex2}). Then, we get directly Lemma~\ref{lem:complex} with the composition $\pcomp = \pcompb \circ 
\pcompa$.

\begin{lemma}
	\label{lem:complex1}
	In a BADC $\Dm_e$, $\compa {\allzero} = ((\one\zero)^{\frac{n}{2}},\one^m)$ holds and takes at most $(n-1)(n+m-2)$	updates. 
\end{lemma}

\begin{proof}
	In this proof, we show that invariant \textbf{inv($i$)} defined as ``at the end of the {$i$th} iteration of the loop, the 
	configuration is $\begin{cases} 
		(\one^{n-i-1}(\one\zero)^{\frac{i+1}{2}},\one^m) & \text{if } i \text{ is odd}\\
		(\zero^{n-i-1}(\zero\one)^{\frac{i}{2}}\zero,\zero^m) & \text{otherwise}
	\end{cases}$'' holds for all $i \in \{1, \ldots, n-1\}$. Notice that we denote by arrow $x \stackrel{k}{\rightsquigarrow} x'$ the 
	transformation of $x$ into $x'$ by the execution of line~$k$ of the sequence considered (\emph{i.e.} $\pcompa$ here).\smallskip
	
	\noindent At the initialisation step ($i=1$), we have:\\
	\begin{tabular}{m{14mm}m{5mm}m{80mm}}
		$\allzero$ & $\stackrel{02}{\rightsquigarrow}$ & $(\one\zero^{n-1},\one\zero^{m-1})$\\
		& $\stackrel{03}{\rightsquigarrow}$ & $(\one^{n-1}\zero,\one\zero^{m-1})$\\ 
		& $\stackrel{04}{\rightsquigarrow}$ & $(\one^{n-1}\zero,\one^{m})$ $=$ $(\one^{n-2}(\one\zero),\one^{m})$\\
	\end{tabular} 
	
	\noindent and \textbf{inv($1$)} is true.\smallskip
	
	\noindent At the maintenance steps, we have:
	
	\begin{itemize}
	\item if $i \equiv 0 \mod 2$, at the beginning of the iteration, the configuration comes from iteration $i-1$ (that is odd) and is 
	consequently $(\one^{n-(i-1)-1}(\one\zero)^{\frac{(i-1)+1}{2}},\one^m)$. Thus we have:\\[1mm]
		\begin{tabular}{m{26mm}m{5mm}m{80mm}}
			$(\one^{n-i}(\one\zero)^{\frac{i}{2}},\one^m)$ & $\stackrel{02}{\rightsquigarrow}$ & 
			$(\zero\one^{n-i-1}(\one\zero)^{\frac{i}{2}},\zero\one^{m-1})$\\
			& $\stackrel{03}{\rightsquigarrow}$ & $(\zero^{n-i}(\one\zero)^{\frac{i}{2}}, \zero\one^{m-1})$ $=$ 
			$(\zero^{n-i-1}(\zero\one)^{\frac{i}{2}}\zero,\zero\one^{m-1})$\\
			& $\stackrel{04}{\rightsquigarrow}$ & $(\zero^{n-i-1}(\zero\one)^{\frac{i}{2}}\zero,\zero^{m})$\\
		\end{tabular}
	\item if $i \equiv 1 \mod 2$, at the beginning of the iteration, the configuration comes from iteration $i-1$ (that is even) and is 
	consequently $(\zero^{n-(i-1)-1}(\zero\one)^{\frac{i-1}{2}}\zero,\zero^m)$. Thus we have:\\[1mm]
		\begin{tabular}{m{30.5mm}m{5mm}m{80mm}}
			$(\zero^{n-i}(\zero\one)^{\frac{i-1}{2}}\zero,\zero^m)$ & $\stackrel{02}{\rightsquigarrow}$ &
			$(\one\zero^{n-i-1}(\zero\one)^{\frac{i}{2}}\zero,\one\zero^{m-1})$\\
			& $\stackrel{03}{\rightsquigarrow}$ & $(\one^{n-i}(\zero\one)^{\frac{i-1}{2}}\zero,\one\zero^{m-1})$\\
			& & \qquad$=$ $(\one^{n-i-1}(\one\zero)^{\frac{i}{2}+1},\one\zero^{m-1})$\\
			& $\stackrel{04}{\rightsquigarrow}$ & $(\one^{n-i-1}(\one\zero)^{\frac{i+1}{2}},\one^{m})$\\[-2mm]
		\end{tabular}
	\end{itemize}
	
	\noindent and \textbf{inv($i$)}, $2 \leq i \leq n-1$, still holds.\smallskip
	
	\noindent At the termination step, since $\Cone$ size is even by hypothesis, $n-1$ is odd and \textbf{inv($n-1$)} holds.\smallskip
	
	Thus $\allzero$ is transformed into $(\one^{n-(n-1)-1}(\one\zero)^{\frac{n-1+1}{2}},\one^m) = ((\one\zero)^{\frac{n}
	{2}},\one^m)$, which is the expected result. Moreover, remark that the number of effective updates made by $\pcompa$ is at most $(n-1)
	(n+m-2)$.\qed
\end{proof}

\begin{lemma}
	\label{lem:complex2}
	In a BADC $\Dm_e$, $\compb {((\one\zero)^{\frac{n}{2}},\one^m)} = \altone$ holds and takes at most $(m-2)(n+m-2) + (2m-1)$ updates. 
\end{lemma}

\begin{proof}
	This proof is similar to that of Lemma~\ref{lem:complex1}. Indeed, we show that invariant \textbf{inv($i$)} defined as ``at the end of 
	the {$i$th} iteration of the loop, the configuration is $\begin{cases} 
		((\zero\one)^{\frac{n}{2}},\zero^{m-i-1}(\one\zero)^{\frac{i+1}{2}}) & \text{if } i \text{ is odd}\\
		((\one\zero)^{\frac{n}{2}},\one^{m-i-2}(\one\zero)^{\frac{i}{2}+1}) & \text{otherwise}
	\end{cases}$'' holds for all $i \in \{1, \ldots, m-2\}$.\smallskip 
	
	\noindent Let us first consider lines 1 to 6 of $\pcompb$, before we enter the loop. Since the configuration is 
	$((\one\zero)^{\frac{n}{2}}, \one^m)$ initially, these lines transform it into $((\one\zero)^{\frac{n}{2}}, \one^{m-2}(\one\zero))$ with 
	respect to the following changes:\\
	\begin{tabular}{m{19mm}m{5mm}m{80mm}}
		$((\one\zero)^{\frac{n}{2}}, \one^m)$ & $\stackrel{01}{\rightsquigarrow}$ & 
		$(\zero\zero(\one\zero)^{\frac{n}{2}-1}, \zero\one^{m-1})$\\
		& $\stackrel{02}{\rightsquigarrow}$ & $(\zero\zero(\one\zero)^{\frac{n}{2}-1}, \zero^m)$\\
		& $\stackrel{03}{\rightsquigarrow}$ & $((\one\zero)^{\frac{n}{2}}, \one\zero^{m-1})$\\
		& $\stackrel{04}{\rightsquigarrow}$ & $((\one\zero)^{\frac{n}{2}}, \one^{m-1}\zero)$ $=$ $((\one\zero)^{\frac{n}{2}}, 
			\one^{m-2}(\one\zero))$\\
	\end{tabular}\smallskip
		 
	\noindent Now, at the initialisation step ($i=1$) of the loop, we have:\\
	\begin{tabular}{m{29.5mm}m{5mm}m{80mm}}
		$((\one\zero)^{\frac{n}{2}}, \one^{m-2}(\one\zero))$ & $\stackrel{06}{\rightsquigarrow}$ & 
		$(\one(\one\zero)^{\frac{n}{2}-1}\one, \one^{m-2}(\one\zero))$\\
		& $\stackrel{07}{\rightsquigarrow}$ & $(\zero(\one\zero)^{\frac{n}{2}-1}\one, \zero\one^{m-3}(\one\zero))$ $=$ 
			$((\zero\one)^{\frac{n}{2}}, \zero\one^{m-2}\zero)$\\ 
		& $\stackrel{08}{\rightsquigarrow}$ & $((\zero\one)^{\frac{n}{2}}, \zero^{m-2}(\zero\one))$\\
	\end{tabular} 
	
	\noindent and \textbf{inv($1$)} is true.\smallskip
	
	\noindent At the maintenance steps, we have:

	\begin{itemize}
	\item if $i \equiv 0  \mod 2$:\\[1mm]
		\begin{tabular}{m{31.5mm}m{5mm}m{100mm}}
			$((\zero\one)^{\frac{n}{2}},\zero^{m-i}(\one\zero)^{\frac{i}{2}})$ & $\stackrel{06}{\rightsquigarrow}$ & 
			$(\zero(\zero\one)^{\frac{n}{2}-1}\zero, \zero^{m-i}(\one\zero)^{\frac{i}{2}})$\\
			& $\stackrel{07}{\rightsquigarrow}$ & $(\one(\zero\one)^{\frac{n}{2}-1}\zero, \one\zero^{m-i-1}(\one\zero)^{\frac{i}{2}})$\\
			& & \qquad$=$ $((\one\zero)^{\frac{n}{2}}, \one\zero^{m-i-1}(\one\zero)^{\frac{i}{2}})$\\
			& $\stackrel{08}{\rightsquigarrow}$ & $((\one\zero)^{\frac{n}{2}}, \one^{m-i-2}(\one\zero)^{\frac{i}{2}+1})$\\
		\end{tabular}
	\item if $i \equiv 1 \mod 2$:\\[1mm]
		\begin{tabular}{m{37.5mm}m{5mm}m{100mm}}
			$((\one\zero)^{\frac{n}{2}}, \one^{m-i-1}(\one\zero)^{\frac{i+1}{2}})$ & $\stackrel{06}{\rightsquigarrow}$ & 
			$(\one(\one\zero)^{\frac{n}{2}-1}\one, \one^{m-i-1}(\one\zero)^{\frac{i+1}{2}})$\\
			& $\stackrel{07}{\rightsquigarrow}$ & $(\zero(\one\zero)^{\frac{n}{2}-1}\one, \zero\one^{m-i-2}(\one\zero)^{\frac{i+1}{2}})$ \\
			& & \qquad$=$ $((\zero\one)^{\frac{n}{2}}, \zero\one^{m-i-2}(\one\zero)^{\frac{i+1}{2}})$\\
			& $\stackrel{08}{\rightsquigarrow}$ & $((\zero\one)^{\frac{n}{2}}, \zero^{m-i-1}(\one\zero)^{\frac{i+1}{2}})$\\[-2mm]
		\end{tabular}
	\end{itemize}
	
	\noindent and \textbf{inv($i$)}, $2 \leq i \leq m-2$, still holds.\smallskip
	
	\noindent At the termination step, since $\Ctwo$ size is even by hypothesis, $m-2$ is even and \textbf{inv($m-2$)} holds.\smallskip
	
	Thus $((\one\zero)^{\frac{n}{2}}, \one^{m-(m-2)-2} (\one\zero)^{\frac{m-2}{2}+1}) = ((\one\zero)^{\frac{n}{2}},  
	(\one\zero)^{\frac{m}{2}})$ is indeed reached by $((\one\zero)^{\frac{n}{2}}, \one^m)$, which is the expected result. Moreover, 
	remark that the number of effective updates made by $\pcompb$ is at most $(m-2)(n+m-2) + (2m-1)$.\qed
\end{proof}

\begin{lemma}
	\label{lem:complex}
	In a BADC $\Dm_e$, $\comp {\allzero} = \altone$ holds and takes at most $(n + m)^2 - 5 (n - 1) - 3m$ updates.
\end{lemma}

\begin{proof}
	The proof that $\comp {\allzero} = \altone$ holds directly derives from Lemmas~\ref{lem:complex1} and~\ref{lem:complex2}. Concerning the 
	number of updates that are needed in the worst case, it suffices to add the maximum number of updates done by $\pcompa$ and $\pcompb$. 
	So, we have:
	\begin{multline*}
		(n-1)(n+m-2) + (m-2)(n+m-2) + (2m-1)\\
		= (n^2 + nm - 3n - m + 2) + (m^2 + nm - 2n - 2m + 4) + (2m -1)\\
		= n^2 + m^2 + 2nm - 5n - 3m + 5\\
		= (n + m)^2 - 5 (n - 1) - 3m\text{,}
	\end{multline*}
	
	\noindent which is the expected result.\qed
\end{proof}

Point~3 is developed in Lemma~\ref{lem:copyPair}, in which we make particular use of $\pcopp$ (cf. Table~\ref{tab:algo0}). 

\begin{lemma}
	\label{lem:copyPair}
	In a BADC $\Dm_e$, for any $x'$, $\copp {\altone} {x'}$ transforms configuration $\altone$ into $x'$ in at most $3(n + m - 4) - 1$ 
	updates. 
\end{lemma}

\begin{proof}
	Let us consider $\pcopp$ where $x = \altone$ and $x'$ is an arbitrary configuration. The proof is done by considering two cases, that of 
	$x_0 \neq x'_0$ and that of $x_0 = x'_0$. The general idea of this proof is to show that $\pcop$ allows to find a sequence from $x$ to 
	$x'$ if $c$ and $c'$ are at the same state. Obviously, in the first case, $x$ needs to be transformed for $\pcop$ to apply correctly. 
	That is what is done in lines~1-5 of $\pcopp$ that transform $x$ into the other most expressive configuration.\smallskip

	So, let us focus on the case $x_0 \neq x'_0$, which means that $c'_0 = 0$, and the transformations that are performed on it by lines~1-5. 
	We have:\\
	\begin{tabular}{m{24.5mm}m{5mm}m{80mm}}
		$\altone$ & $\stackrel{02}{\rightsquigarrow}$ & $(\one(\one\zero)^{\frac{n}{2}-1}\one, (\one\zero)^{\frac{m}{2}})$\\
		& $\stackrel{03}{\rightsquigarrow}$ & $(\one(\one\zero)^{\frac{n}{2}-1}\one, \one(\zero\one)^{\frac{m}{2}-1}\one)$\\
		& $\stackrel{04}{\rightsquigarrow}$ & $(\zero(\one\zero)^{\frac{n}{2}-1}\one, \zero(\zero\one)^{\frac{m}{2}-1}\one)$ $=$ 
			$\altzero$\\
	\end{tabular}\smallskip
	
	On this basis, at line~6, just before $\pcop$ is executed, $c$ and $c'$ necessarily have the same state and Property~1 of 
	Lemma~\ref{lem:copy} holds for $x$ and can thus be applied for ending the proof. Also, notice that, in the worst case, every 
	automaton of $\Cone$ and $\Ctwo$ are updated before the execution of $\pcop$. Thus, this sequence takes at most $3(n + m - 4) - 1$ 
	updates.\qed
\end{proof}

By combining Lemmas~\ref{lem:simply},~\ref{lem:complex} and~\ref{lem:copyPair}, for all configurations $x$ and $x'$, the composition 
$\copp {\comp {\simp {x}}} {x'} = x'$ holds, which shows that there exists a unique attractor of size $2^{n+m-1}$. From this is derived the 
following theorem. 

\begin{theorem}
	\label{thm:oneAttEven}
	A BADC $\Dm_e$ admits a unique attractor of size $2^{n+m-1}$. In this stable oscillation, any configuration can be reached by any other 
	one in $O(n^2 + m^2)$. However, some configurations are specific: $\allzero$ and $\allone$ can be reached from any other one in 
	$O(n + m)$, and configurations $\altzero$ and $\altone$ can reach any configuration in $O(n + m)$.
\end{theorem}

Now we show that the bound $O(n^2 + m^2)$ of Theorem~\ref{thm:oneAttEven} above is tight.

\begin{theorem}
	\label{thm:quad}
	Let $x$ be a configuration of a BADC $\Dm_e$. To increase the expressiveness of $x$ by $\delta \in \mathbb{N}$ needs $\Omega(\delta^2)$ 
	updates. 
\end{theorem}

\begin{proof}
	Although this theorem deals with configuration $x$, let us focus with no loss of generality on how the expressiveness of $x^\ell$ can be 
	increased by $\delta$. First, notice that the only way to increase the expressiveness of $\Cone$ needs to use $\sync$. However, to 
	execute two $\sync$ puts $c_0$ at its initial state. So, to be efficient, the two $\sync$ have to be separated by a sequence of updates. 
	Take for instance the following sequence of instructions for $i \in \{1, \ldots, n\}$: $\sync;$ $\clock {\Cone} {1} {i};$ $\sync;$ 
	$\clock {\Cone} {1} {i};$. With this sequence, the second call to $\pclock$ leads to replace all the information created by the first one 
	and contained by automata $c^\ell_1, \ldots, c^\ell_i$. As a consequence, to create $\delta$ new patterns $\zero\one$ in $\Cone$ needs 
	the calls of $\sync$ to be separated by specific sequences which propagate along the cycle the patterns generated by the previous call to 
	$\sync$. Now, since there is at least $\delta$ calls to $\sync$, just after its $i$th call, the pattern has to be propagated at least 
	until automaton $c_{\delta - (i - 1)}$. Thus, the $i$th call to $\sync$ has to be followed by at least $\delta - (i - 1)$ updates in 
	order to ensure that the pattern is effectively kept. As a result, at the end, to increase the expressiveness of $\Cone$ by $\delta$
	patterns needs $\Omega(\delta^2)$ updates.\qed
\end{proof}

Corollary~\ref{cor:quad} is then directly derived from the two previous theorems, considering that $\delta = \frac{n}{2}$ for $\Cone$ and 
$\delta = \frac{m}{2}$ for $\Ctwo$.

\begin{corollary}
	\label{cor:quad}
	In a BADC $\Dm_e$, to reach $\altone$ from $\allzero$ requires $\Theta(n^2 + m^2)$ steps. 
\end{corollary}

\subsubsection{At least one cycle is odd} Like BADCs $\Dm_e$, BADCs $\Dm_o$ admit only one attractor but contrary to the latter, they also 
admit a set $I$ of specific non-recurrent configurations, from which updates are ``irreversible'' (\emph{i.e.} configurations of $I$ are not 
accessible). In the sequel, abusing language, these configurations are said to be \emph{irreversible}. Lemma~\ref{lem:inacDouble} below shows 
the irreversibility of some configurations.

\begin{lemma}
	\label{lem:inacDouble}
	Let us consider a BADC $\Dm_o$. The following properties hold:
	\begin{enumerate}
	\item If $\mathscr{C}^s$, $s \in \{\ell, r\}$, is of odd size $k > 1$, then configuration $x$ such that $x^s = ((\one\zero)^{\frac{k-1}
		{2}}\one)$ is irreversible.
	\item If both $\Cone$ and $\Ctwo$ are of odd sizes $n > 1$ and $m > 1$, then configuration $((\zero\one)^{\frac{n-1}{2}} \zero,
		(\zero\one)^{\frac{m-1}{2}} \zero)$ is irreversible.  
	\end{enumerate}
\end{lemma}

\begin{proof}
	Let us consider an arbitrary BADC $\Dm_o$. The proof is divided into two parts. First, without loss of generality, we show that when $n$ 
	is odd, configurations $((\one\zero)^{\frac{n-1}{2}}\one, \centerdot)$ (where $\centerdot$ denotes any configuration $x^r$ of $\Ctwo$) 
	are irreversible. Second, we prove the irreversibility of configuration 
	$((\zero\one)^{\frac{n-1}{2}}\zero,(\zero\one)^{\frac{m-1}{2}}\zero)$ when $n$ and $m$ are odd.\medskip
	
	\indent -- \emph{Irreversibility of $((\one\zero)^{\frac{n-1}{2}}\one,\centerdot)$}\smallskip
	
	In order to simplify this part, we consider a BAN composed of a negative cycle $\mathscr{C}$ of odd size $n$ and of an automaton 
	$c^\star$. This BAN is defined by the following $n + 1$ local transition functions:\\
	\centerline{$f_{c_0}(x) = \neg x_{n-1} \land x_{c^\star}$, $f_{c^\star}(x) = \neg x_{c^\star}$ and $\forall i \in \{1, \ldots, n-1\},\ 
	f_{c_i}(x) = x_{i-1}$.}\vspace*{1mm}
	The configurations of this BAN will be denoted by $(x_{c_0}\ldots x_{c_{n-1}}, x_{c^\star})$. Remark that the idea underlying automaton 
	$c^\star$ is to represent an atomic element that acts on $\mathscr{C}$, as another cycle should do. However, this interacting element is 
	more expressive than a cycle since its state switches as soon as it is updated (it plays the role of an oscillator). In fact, in the 
	context of BADCs, $c^r_{m-1}$ plays the role of $c^\star$. However, the effective updates of $c^r_{m-1}$ are clearly more restricted than 
	that of $c^\star$ since they directly depend on the configuration of $\Ctwo$ and indirectly on that of $\Cone$.\smallskip
	
	Now, let us consider configurations $x = ((\one\zero)^{\frac{n-1}{2}}\one), \centerdot)$ and $x'$ obtained by executing $\update 
	{c_i}$ on $x$, for any $i \in \{0, \ldots, n-1\}$. Given the nature of $x$, notice that $x' \neq \bar{x}^i$. In order to prove the 
	result, we have to show that there are no sequences to reach $x$ from $x'$. To do so, we reason by contradiction.\smallskip
	
	Let us suppose that there exists a sequence $\sigma$ composed of $\pupdate$ and $\sync$ instructions that transforms $x'$ into $x$. Since 
	$x'$ is made from $x$ by updating $c_i$, in order to get back to $x$, $c_i$ in $x'$ needs to be switched again in order to have $x'_i = 
	x_i$, which implies that $c_{i-1}$ has to be updated. Thus, sequence $\sigma$ contains at least one $\update {c_{i-1}}$. Furthermore, 
	if we want $c_{i-1}$ to change its state, $c_{i-2}$ has to be switched, and so on. Thus, by iterating this argument, $c_0$ has to be 
	switched too, which implies that $\sigma$ necessarily contains at least one $\sync$, that leads $c_0$ to take state $\zero$. From this, 
	while $x'_{n-1} = \one$, if we want $c_0$ to get back to state $\one$, $c_{n-1}$ has to switch to state $\zero$, which imposes that 
	$\sigma$ contains also at least one $\update {c_{n-1}}$. Furthermore, if $i = 0$ (\emph{i.e.} $x' = \bar{x}^0$) or $i = n-1$ (\emph{i.e.} 
	$x' = \bar{x}^{n-1}$), it is obvious that $\sigma$ leads the process to get back to $x$ to a punctual configuration where the state of 
	$c_{n-1}$ is $\zero$. As a consequence, any sequence that transforms $x'$ into $x$ needs to temporarily transform $x'$ into a 
	configuration in which $c_{n-1}$ is at state $\zero$. Let us focus essentially on $c_{n-1}$ in order to determine properties that need to 
	hold in $\sigma$. It is easy to understand that the two following properties are necessary:
	\begin{enumerate}
	\item[(a)] $\sigma$ must contain at least one $\update {c_{n-1}}$.
	\item[(b)] Any time $\sigma$ passes through a configuration $y$ where $y_0 = \zero$ and $y_{n-1} = \one$, then the sequence of 
		transformations made by $\sigma$ on $y$ to get back to $x$ contains at least two occurrences of $\update {c_{n-1}}$.
	\end{enumerate}
	
	Let us now consider the suffix of $\sigma$ that begins within the last occurrence of $\update {c_{n-1}}$. According to (a), such a suffix 
	exists. Furthermore, according to (b), this suffix does not contain the instruction $\sync$. Indeed, if it was the case, it would mean 
	that the suffix would temporarily lead $y$ to be transformed into $y'$ such that $y'_0 = \zero$ since $c_{n-1}$ is at state $\one$. Now, 
	according to (b), in such a case, the suffix should contain two occurrences of $\update {c_{n-1}}$, which contradicts the hypothesis made 
	on the suffix according to which it cannot contain any $\update {c_{n-1}}$.\smallskip
	
	From above, if we denote by $z$ the configuration obtained from $x'$ by applying the instructions of $\sigma$ until the last $\update 
	{c_{n-1}}$ (included), because there are no instructions $\sync$ and no more instructions $\update {c_{n-1}}$ in the suffix, $z_0$ and 
	$z_{n-1}$ have to equal $\one$. Now, for $n > 1$ (\emph{i.e.} $c_0 \neq c_{n-1}$), since the last update of $c_{n-1}$ leads it to state 
	$\one$, this means that $z_{n-2} = \one$. Now, from that, the only possible way to get back to $x$ from $z$ would be to put $c_{n-i}$ to 
	state $\zero$ if $i$ is even and to state $\one$ otherwise. But this implies $c_0$ to be at state $\zero$ at a certain step, which would 
	need to call to $\sync$ again, which contradicts the definition of the suffix as stated above. As a result, such sequence $\sigma$ does 
	not exist.\smallskip
	
	We have just shown that configuration $((\one\zero)^{\frac{n-1}{2}}\one,\centerdot)$ cannot be reached, even if an over-expressive 
	element $c^\star$ acting on $c_0$ is considered instead of a second negative cycle. So, $((\one\zero)^{\frac{n-1}{2}}\one,\centerdot)$ is 
	irreversible, and with trivial extensions, we obtained the expected result given in Point~1 of Lemma~\ref{lem:inacDouble}.\medskip
	
	\indent -- \emph{Irreversibility of $((\zero\one)^{\frac{n-1}{2}}\zero,(\zero\one)^{\frac{m-1}{2}}\zero)$}\smallskip
	
	Consider here a BADC $\Dm_o$ where $\Cone$ and $\Ctwo$ are of odd sizes, respectively such that $n > 1$ and $m > 1$. The proof of the 
	irreversibility of $((\zero\one)^{\frac{n-1}{2}}\zero,(\zero\one)^{\frac{m-1}{2}}\zero)$ is done by contradiction too.\smallskip
	
	Let $x'$ be the configuration that results from the update of one automaton in configuration $x = ((\zero\one)^{\frac{n-1}{2}}\zero,
	(\zero\one)^{\frac{m-1}{2}}\zero)$. With the same argument as above, any sequence $\sigma$ that allows to transform $x'$ into $x$ must 
	pass temporarily through a configuration $y$ such that $y_0 = \zero$ and $y^\ell_{n-1} = \one$ or $y^r_{m-1} = \one$. From $y$, in order 
	that $\sigma$ transforms it into $x$, $\sigma$ needs to contain at least one $\update {c_{n-1}}$ or one $\update {c_{m-1}}$ depending on 
	what automaton $c_{n-1}$ or $c_{m-1}$ is at state $\one$ (maybe both). Consider now the last occurrence of $\update {c_{n-1}}$ or 
	$\update {c_{m-1}}$. After that occurrence, $c_{n-1}$ and $c_{m-1}$ are at state $0$. Otherwise, $\sigma$ cannot lead to $x$. From that, 
	we derive that:
	\begin{enumerate}
	\item[\emph{(i)}] the temporary configuration obtained $z \neq x$ is such that $z_0 = 0$, $z^\ell_{n-1} = z^r_{m-1} = 0$ and 
		$z^\ell_{n-2} = 0$ or $z^r_{m-2} = 0$ (maybe both), and
	\item[\emph{(ii)}] no instructions $\sync$ can follow because this would imply the presence of other $\pupdate$ on $c_{n-1}$ or $c_{m-1}$ 
		which contradicts the hypothesis that we focused on the last occurrence of such an $\pupdate$.
	\end{enumerate}

	Now, consider for instance that $z^\ell_{n-2} = 0$, with no loss of generality. From that, as above, the only possible way to get back to 
	$x$ from $z$ would be to put $c_{n-i}$ to state $\zero$ if $i$ is even and to state $\one$ otherwise. But this implies $c_0$ to be at 
	state $\zero$ at a certain step, which would need to call to $\sync$ for $c_0$ to switch to state $\one$. A contradiction with 
	\emph{(ii)} appears, which shows that such sequence $\sigma$ does not exist and, consequently, that $((\zero\one)^{\frac{n-1}{2}}\zero,
	(\zero\one)^{\frac{m-1}{2}}\zero)$ is irreversible, as stated in Point~2 of Lemma~\ref{lem:inacDouble}.\qed	
\end{proof}

Let $I$ be the set of irreversible configurations of a BADC $\Dm_o$ given by Lemma~\ref{lem:inacDouble}. Theorem~\ref{thm:sizeof} below 
proves that $I$ contains in fact all the irreversible configurations and, from this set, generalises Theorem~\ref{thm:oneAttEven} for any 
sort of negative BADCs. Notice that the complexity bounds remain valid. They are consequently not given again. 

\begin{theorem}
	\label{thm:sizeof}
	Let $\alpha: \mathbb{N} \rightarrow \{0,1\}$ with $\alpha(k) = \begin{cases} 
		0 & \text{if } k = 0 \text{ or } k \equiv 1 \mod 2\\
		1 & \text{otherwise}
	\end{cases}$
	 Any negative BADC $\Dm$ admits one attractor of size $2^{n + m - 1} - |I|$, where $|I| = \alpha(n - 1) \times 2^{m - 1}  -  
	\alpha(m - 1) \times 2^{n - 1}$.  
\end{theorem}

\begin{proof}
	In this proof, we focus on BADC $\Dm_o$ since the case of negative BADCs composed of two cycles of even sizes has been treated 
	previously. Let us begin by showing that $\Dm_o$ admits only one attractor that contains all the configurations except those belonging to 
	$I$.\smallskip
	
	First, we have to prove that any configuration $x$ of $\Dm_o$ can be transformed into the lowest expressive configuration $\allzero$. 
	Following the proof of Lemma~\ref{lem:simply}, we get $\simp {x} = \allzero$ and the related complexity still holds.\smallskip
	
	Second, let us focus on the increase of the expressiveness of configurations. To do so, let us consider two cases: (a) only one cycle is 
	of odd size and we consider that it is $\Cone$ with no loss of generality; (b) both cycles are of odd sizes. According to both these 
	cases, we have:
	\begin{enumerate}
	\item[(a)] $x\zero = ((\zero\one)^{\frac{n-1}{2}}\zero, (\zero\one)^{\frac{m}{2}})$ and $x\one = ((\one\zero)^{\frac{n-1}{2}}\zero, 
		(\one\zero)^{\frac{m}{2}})$ are two of the three most expressive configurations that do not belong to $I$ (the third one is 
		$\bar{x\one}$ that has not to be taken into account because the results for $x\one$ extend to it directly). Notice that $x\zero$ can 
		be transformed into $x\one$ by means of sequence $\sigma_a = \shift {\Cone};$ $\shift {\Ctwo};$ $\update {c_{n-1}};$ $\sync;$. 
		Conversely, the sequence $\sigma'_a = \shift {\Cone};$ $\shift {\Ctwo};$ $\sync;$ allows to reach $x\zero$ from $x\one$.
	\item[(b)] $x\zero = ((\zero\one)^{\frac{n-1}{2}}\one, (\zero\one)^{\frac{m-1}{2}}\zero)$ and $x\one = ((\one\zero)^{\frac{n-1}{2}}\zero, 
		(\one\zero)^{\frac{m-1}{2}}\zero)$ are two of the three most expressive configurations that do not belong to $I$ (the third one is 
		$\bar{x\one}$ and is not considered for the same reason as above). In this case, $x\zero$ can be transformed into $x\one$ by means 
		of sequence $\sigma_b = \shift {\Cone};$ $\shift {\Ctwo};$ $\update {c_{n-1}};$ $\update {c_{m-1}};$ $\sync;$. Also, the 
		sequence $\sigma'_b = \shift {\Cone};$ $\shift {\Ctwo};$ $\update {c_{n-1}};$ $\sync;$ allows to reach $x\zero$ from $x\one$.
	\end{enumerate}
	
	From the reasoning given in the proofs of Lemmas~\ref{lem:complex1},~\ref{lem:complex2} and~\ref{lem:complex}, it can be derived that 
	$\comp{\allzero}=x\zero$, which shows together with $\sigma_a$ and $\sigma_b$ the accessibility of the most expressive configurations 
	from the least expressive ones in any case. Notice that the bound $\Theta(n^2 + m^2)$ remains valid in this case.\smallskip
	
	Third, consider now a new version of $\pcomp$ that takes as parameters a configuration and either $\zero$ or $\one$. More precisely, this 
	new version of $\pcomp$ is defined such that:
	\begin{gather*}
		\comp {z,\zero} = \comp {z} \text{ and}\\
		\comp {z,\one} = \begin{cases}
			\comp {z};\ \sigma_a; & \text{for case (a)}\\
			\comp {z};\ \sigma_b; & \text{for case (b)}
		\end{cases}\text{.}
	\end{gather*}
	Let $x$ and $y$ be two configurations that do not belong to $I$. First, remark that the state of $c_0$ in $\comp 
	{\allzero, y_0}$ equals $y_0$. Consequently, since we have $\comp {\allzero, y_0}_0 = y_0$, thanks to the proofs of Lemmas~\ref{lem:copy} 
	and~\ref{lem:copyPair} and the fact that $\comp {\simp {x}, y_0}$ corresponds to one of the most expressive configurations in any case, 
	$y = \cop {\comp {\simp {x}, y_0}} {y}$ holds. As a result, all the configurations that do not belong to $I$ are recurrent and are 
	reachable from each other, which implies that they compose a unique attractor. Notice also that from this result, we have easily the 
	intermediary result stating that the number of updates to reach any configuration from any of the most expressive configurations is 
	linear. Indeed, since $x\zero$ and $x\one$ can reach each other through the distinct linear sequences $\sigma$s and since we 
	have just shown that they can reach any configuration $y \notin I$ by using $\pcop$, it is direct that the most expressive 
	configurations can transform themselves linearly into any other configuration.\smallskip
	
	To complete this part, by basing ourselves on what has been done until now, let us focus on the cases where either $n$ or $m$ equals $1$. 
	Consider without loss of generality that $m = 1$. We have to distinguish two cases:
	\begin{itemize}
	\item $n \equiv 0 \mod 2$: This case is trivial because no irreversible configurations exist. As a result, such a BADC admits one 
		attractor of size $2^n$ and $I = \emptyset$.
	\item $n \equiv 1 \mod 2$: First, remark that $\sync((\one(\one\zero)^{\frac{n-1}{2}},\one)) = (\zero\one)^{\frac{n-1}{2}}\zero, \zero)$, 
		which is thus not irreversible. Second, according to Lemma~\ref{lem:inacDouble}, $((\one\zero)^{\frac{n-1}{2}}\one, \one)$ 
		is irreversible. Consequently, it is the only one that cannot be reached. As a result, such a BADC admits one attractor of size 
		$2^n-1$ and $I = \{((\one\zero)^{\frac{n-1}{2}}\one, \one)\}$.
	\end{itemize}

	We have proven that a negative BADC $\Dm$ has only one attractor, whatever the cycle parity. However, the size of this 
	attractor depends on the cardinal of $I$, on which we focus from now. Several cases have to be taken into account:
	\begin{itemize}
	\item $n = 1$ or $m = 1$: In this case, as stated just above, if the cycle of size $1$ intersect with a cycle of even size 
		(resp. odd size), there are no irreversible configurations and $|I| = 0$ (resp. there is one irreversible configuration and $|I| = 
		1$). 
	\item $n$ and $m$ are greater than $1$: 
		\begin{itemize}
		\item If both cycles are of even sizes: such BADCs admits a unique attractor of size $2^{n + m - 1}$ and $|I| = 0$ (cf. 
			Theorem~\ref{thm:oneAttEven}).
		\item If only one of the cycles is of odd size: if this cycle is $\Cone$ (resp. $\Ctwo$) then the configurations of the form 
			$((\one\zero)^{\frac{n-1}{2}}\one, \centerdot)$ (resp. $(\centerdot, (\one\zero)^{\frac{n-1}{2}}\one)$) are irreversible and 
			$|I| = 2^{m-1}$ (resp. $|I| = 2^{n-1}$).
		\item If both cycles are of odd sizes: $I$ is in this case composed of configurations of the forms $((\one\zero)^{\frac{n-1}{2}}\one, 
			\centerdot)$ and $(\centerdot, (\one\zero)^{\frac{n-1}{2}}\one)$, and configuration $((\zero\one)^{\frac{n}{2}}\zero,
			(\zero\one)^{\frac{m}{2}}\zero)$. That means that $|I| = 2^{m - 1} + 2^{n - 1}$.
		\end{itemize}
	\end{itemize}
	
	From this, we derive the following generalisation that states that any negative BADC $\Dm$ admits a unique attractor, a stable 
	oscillation of size 
	\begin{equation*}
		2^{n + m - 1} - \left( \alpha(n - 1) \times 2^{m - 1}  +  \alpha(m - 1) \times 2^{n - 1} \right)\text{,}
	\end{equation*}
	where $\alpha: \mathbb{N} \rightarrow \{0,1\}$ with $\alpha(k) = \begin{cases} 
		0 & \text{if } k = 0 \text{ or } k \equiv 1 \mod 2\\
		1 & \text{otherwise}
	\end{cases}$.\qed
\end{proof}

\section{Conclusion and perspectives}
\label{sec:conc}

This paper followed the lines of~\cite{Demongeot2012,Noual2012b} and focused on the dynamical properties of BADCs subjected to the 
asynchronous updating mode. Again, the focus on BADCs is explained by the fact that although cycles have been known to be the engines of 
complexity in interaction networks since the 1980's, their influence on network dynamics is not really known, contrary to the common beliefs. 
This needs to be changed if we want to understand precisely interaction network complexity. However, because of the intrinsic difficulties to 
bring such studies in general frameworks (in general BANs for instance), we needed to restrain the spectrum of intersections considered to 
the ``simplest'' kinds, the tangential ones. In this setting, our contribution was twofold: \emph{(i)} we gave a complete characterisation of 
the dynamical behaviour of asynchronous BADCs by means of \emph{(ii)} new algorithmic tools that bring a new way to view updates in networks 
and a nice understanding of how information is relayed. Obviously, these tools have been built for our purpose and their use is consequently 
limited. Nevertheless, remark that they can be applied almost directly in some more complex networks, in particular those with tangential 
cycle intersections, such as flower graphs for which they will allow to provide characterisation results regarding their behaviours that will 
generalise the existence results given in~\cite{Didier2012}. Furthermore, another perspective would consist in adapting these tools in order 
them to apply to more complex intersections. Beyond the dynamical aspects, notice that the algorithmic tools owe the benefits to represent 
concisely long sequences of updates. About this abstraction, we would like to understand to what extent we can characterise network 
architectures when update sequences (that represent only pieces of dynamics) are given. For instance, the latter could be very useful to find 
networks of specific dynamics complexity classes (in terms of convergence time for instance, or even in terms of number of attractors). To 
finish, this work together with that of~\cite{Noual2012b} (and the differences they present) raises once again the matter of the fundamental 
differences between synchronism and asynchronism whose study deserves to be pursued. 

\bibliography{asyncBADC}

\end{document}